\documentclass[onecolumn,longbibliography,superscriptaddress,doublespace]{revtex4-2}

\usepackage{graphicx}
\usepackage{setspace,color,appendix,physics,wasysym,xcolor,xparse}
\usepackage[margin=1.4in]{geometry}
\usepackage{appendix} 
\usepackage[english]{babel}
\usepackage{amsmath,amssymb,amsfonts,mathtools}
\usepackage{bm,bbm,euscript,braket,esint}
\definecolor{brickred}{rgb}{0.8, 0.25, 0.33}
\definecolor{celestialblue}{rgb}{0.29, 0.59, 0.82}
\definecolor{cornflowerblue}{rgb}{0.39, 0.58, 0.93}
\definecolor{denim}{rgb}{0.08, 0.38, 0.74}
\definecolor{armygreen}{rgb}{0.29, 0.33, 0.13}
\definecolor{cardinal}{rgb}{0.77, 0.12, 0.23}
\definecolor{carnelian}{rgb}{0.7, 0.11, 0.11}
\definecolor{armygreen}{rgb}{0.29, 0.33, 0.13}

\usepackage[colorlinks=true, urlcolor=black, citecolor=blue, linkcolor=red, citebordercolor={1 0 0}, linkbordercolor={0 0 1}]{hyperref}

\newtheorem{Theorem}{Theorem}
\newtheorem{lemma}{Lemma}
\newtheorem{definition}{Definition}
\newenvironment{proof}{{\bf Proof:}}{\hfill$\square$}

\usepackage{verbatim,multirow}
\usepackage[utf8]{inputenc}
\usepackage[T1]{fontenc}
\usepackage{lmodern}
\usepackage{hyperref}
\hypersetup{colorlinks,linkcolor={red},citecolor={blue},urlcolor={blue}}

\usepackage[utf8]{inputenc}
\usepackage[T1]{fontenc}
\usepackage{lmodern}

\begin{document}
	\title{On the eternal non-Markovianity of non-unital quantum channels}
	
	\author{Shrikant Utagi}
	\email{shrikant.phys@gmail.com}
	
	\affiliation{Department of Physics, Indian Institute of Technology Madras, Chennai - 600036, India.}
	
	\author{Subhashish Banerjee}
	\email{subhashish@iitj.ac.in}
	
	\affiliation{Indian Institute of Technology Jodhpur, Jodhpur - 342037,  India.}
	
	\author{R. Srikanth}
	\email{srik@ppsir.res.in}
	\affiliation{Theoretical Sciences Division, Poornaprajna    Institute   of   Scientific
		Research, Bengaluru -  562164, India.}

\begin{abstract}
	The eternally non-Markovian Pauli channel is an example of a unital channel characterized by a negative decay rate for all time $t>0$. Here we consider the problem of constructing an analogous non-unital channel, and show in particular that a $d$-dimensional generalized amplitude damping (GAD) channel cannot be eternally non-Markovian when the non-Markovianity originates solely from the non-unital part of the channel. We study specific ramifications of this result for qubit GAD. Specifically, we construct a quasi-eternally non-Markovian qubit GAD channel, characterized by a time $t^\ast > 0$, such that the channel is non-Markovian  only and for all time $t > t^\ast$. We further point out that our negative result for the qudit GAD channel, namely the impossibility of the eternal non-Markovian property, does not hold for a general qubit or higher-dimensional non-unital channel. 
\end{abstract}

\maketitle



\section{Introduction}
Open quantum systems is the study of the dynamics of the system of interest taking into account the effect of its ambient environment and an important aspect of this is the study of non-Markovian (NM) dynamics \cite{breuer2002theory,nielsen2002quantum,banerjee2018open}, which has gained increasing attention recently in the quantum information theory community \cite{RHP14,breuer2016colloquium,vega2017dynamics,shrikant2023quantum,kumar2018non,naikoo2020non,paulson2021hierarchy}. A number of definitions of non-Markovianity have been proposed namely, based on distinguishability \cite{breuer2009measure},  completely positive (CP-) divisibility, fidelity \cite{rajagopal2010kraus}, negativity of the decay rate \cite{RHP10,hall2014canonical}, quantum temporal correlations \cite{chen2016quantifying,shrikant2021quantum} and a more general definition based on deviation from semigroup structure \cite{shrikant2020temporal}, to name a few. 

The notion of quantum non-Markovianity has always been debated, see for instance Ref.~ \cite{li2018concepts, pollock2018non, pollock2018operational, budini2018quantum, milz2019completely, shrikant2023quantum}. Its revived interest is due to the advent of newer experiments with finer engineered baths, which highlight the fact that there exists no universal definition or criterion which faithfully characterizes, detects or measures quantum non-Markovianity. Most of the developments focused on the two-time dynamical maps and the divisibility criterion based on the two-time correlation functions in determining the Markovian status of open system dynamics, See Ref. \cite{RHP14, breuer2016colloquium, vega2017dynamics}. A stronger notion of quantum Markovianity has recently been developed, that is based on an operational point of view \cite{pollock2018operational, pollock2018non, milz2019completely, budini2018quantum,li2018concepts, budini2019conditional} and eschews the use of the concept of the not-completely positive (NCP) intermediate map. Moreover, this approach meets the status of necessity and sufficiency of characterizing non-Markovianity when multi-time correlations correlations are taken into account. For the present purpose of this work, we restrict ourselves to the divisibility criteria based on the two-time correlations and two-time dynamical map formalism. In particular, we will be concerned with the eternal CP-indivisibility and P-indivisibility of channels whose non-Markovianity originates solely from the non-unital part, which we definite shortly. 

An open system evolution is generally given by the famous Gorini-Kossakowski-Lindblad (GKSL)-like master equation \cite{lindblad1976,sudarshan_stochastic_1961}:
\begin{align}
	\Dot{\rho} = \mathcal{L}(t)[\rho] = - i [H,\rho] - \sum_j \gamma_j(t) \bigg(A_j(t)\rho A_j^\dagger(t) - \frac{1}{2}\{A_j^\dagger(t)A_j(t), \rho \} \bigg),
	\label{eq:master}
\end{align}
where $\gamma_j(t)$ are the canonical decay rates \cite{hall2014canonical} and $A_j(t)$ are the jump operators. Assuming that the system state $\rho_S(t=0)$ and environment initially \textit{fixed} in the state $\rho_E = \ket{e_0}\bra{e_0}$ is a product i.e., $\rho_{SE} = \rho_S(0) \otimes \rho_E $, one can deduce a completely positive trace preserving (CPTP) dynamical map on the system state via the operator-sum form \cite{sudarshan_stochastic_1961,nielsen2002quantum}: 
$$\mathcal{E}(t)[\rho] = \sum_j \bra{e_j} U (\rho_S(0) \otimes \ket{e_0}\bra{e_0}) U^\dagger \ket{e_j} =  \sum_j E_j(t) \rho_S E_j^\dagger(t),$$ 
where $K_j(t) \equiv  \bra{e_j} U \ket{e_0}$ are called the Kraus operators, where $U$ is joint system-environment unitary. The map $\mathcal{E}$ is in fact the solution to the master equation $\Dot{\mathcal{E}} = \mathcal{L}\mathcal{E}$, which, for all $t \ge s \ge t_0$, is given by $\mathcal{E}(t,t_0) = \mathcal{T} \exp\{\int_{t_0}^{t} ds \mathcal{L}(s) \}$, where $\mathcal{T} $ is the time-ordering operator. However, it is well-known that when the initial factorization assumption is relaxed, the ensuing dynamical map need not be CP \cite{pechukas1994reduced, alicki1995comment, pechukas1995pechukas}. Moreover, if a map could be written down, one must sacrifice either complete positivity or linearity. Therefore, we stress that the dynamical maps defined throughout this work are obtained based on the initial factorization assumption.

Now, we are in a position to define the concept of divisibility of a map.

A CPTP qubit map $\mathcal{E}(t_2,t_0)$ (a qubit quantum channel) is said to be CP-divisible if for arbitrary intermediate time $t_1$ it can be written as 
\begin{align}
	\mathcal{E}(t_2,t_0) = \mathcal{E}(t_2,t_1)\mathcal{E}(t_1,t_0), \label{eq:divi}
\end{align}
such that $\mathcal{E}(t_2,t_1)$ is a CP-map. The Choi-Jamiolkowski matrix \cite{choi1975completely} or B-matrix \cite{sudarshan1961stochastic} for a map $\mathcal{E}(t_2,t_1)$ is given by \begin{align}
	\mathcal{\chi} = (\mathcal{E}(t_2,t_1) \otimes I)[\ket{\Psi}\!\bra{\Psi}],
	\label{eq:choi}
\end{align} where $\ket{\Psi} = \frac{1}{\sqrt{2}}(\ket{00}+\ket{11})$ is a maximally entangled state. When $\chi$ is negative, the channel $\mathcal{E}(t_2,t_0)$ is CP-indivisible,  in the sense that it is a CPTP map that cannot be written as composition of CP maps as in Eq.~(\ref{eq:divi}).  For bijective maps,  the above requirement is equivalent to the following condition for CP-divisibility 
\begin{align}
	\frac{d}{dt}\Vert [\mathbb{I}_{d+1} \otimes \mathcal{E}(t)](\rho_1-\rho_2) \Vert_1 \le 0
	\label{eq:CPindiv}
\end{align}
being always satisfied for given any two density operators $\rho_1$ and $\rho_2$ in the space of bounded operators of the bipartite system $\mathcal{B}(H_{d+1} \otimes H_d)$ \cite{chruscinski2018divisibility}. Here, $\Vert A \Vert_1 = {\rm Tr}(\sqrt{A^\dagger A})$ is the trace norm of an operator $A$. 
The channel $\mathcal{E}(t_2,t_0)$ satisfies P-divisibility if 
\begin{align}
	\frac{d}{dt}\Vert \mathcal{E}(t)[\rho_1 - \rho_2] \Vert_1 \le 0
	\label{eq:Pindiv}
\end{align}
for any initial pair of states $\rho_1$ and $\rho_2$.  Note that when the condition (\ref{eq:Pindiv}) breaks down, the intermediate map $\mathcal{E}(t_2,t_1)$ would no longer be positive, in the sense that it outputs a negative state. However, the full map $\mathcal{E}(t_2,t_0)$ is still CPTP but no longer P-divisible. 

The notion of an ``eternally non-Markovian'' (ENM) Pauli channels was proposed in \cite{hall2014canonical}  via the canonical master equation given by
\begin{align}
	\mathcal{L}(t)[\rho] = \sum_{j=1}^{3} \gamma_j(t) (\sigma_j \rho \sigma_j - \rho )    
	\label{eq:eternalmaster}
\end{align}
where $\sigma_j$ are the Pauli operators, and the canonical decay rates $\gamma_j(t)$ are given by $\gamma_1=\gamma_2=\frac{c}{2}$ and $\gamma_3(t) = - \frac{c}{2}\tanh t $, where $c$ is some real constant. This particular dynamics corresponds to CPTP map $\mathcal{E}(t)\rho = \sum_{j=1}^{3} E_j \rho E_j^\dagger$, with the Kraus operators given by 
\begin{equation}
	E_1 = \sqrt{2k}\sigma_I,~~ E_2 = \sqrt{k}\sigma_X; ~~E_3 = \sqrt{k}\sigma_y
	\label{eq:ENM}
\end{equation} 
where 
\begin{equation}
	k \equiv \frac{1}{4}(1-e^{-c t}).
	\label{eq:k}
\end{equation}
For Pauli channels, it is straightforward to arrive at Eq.~(\ref{eq:ENM}) from Eq.~(\ref{eq:eternalmaster}) from the so-called $H$-matrix method \cite{chruscinski2013non,chruscinski2015non}.  This example was later shown to arise under the mixing of two Pauli semigroups \cite{wudarski2016markovian}, providing an example of the fact that the space of CP-divisible Pauli channels is not convex \cite{wolf2008assessing}. This is an instance in which the criterion Eq. (\ref{eq:Pindiv}) for witnessing non-Markovianity  fails \cite{hall2014canonical}. Interestingly, in \cite{megier2017eternal} the authors studied whole families of eternal and quasi-eternal unital channels. ENM evolution arising from a convex combination of Markovian semigroups has recently been extended to the context of generalized Pauli channels in dimension $k$ \cite{siudzinska2020quantum}, where it is shown that $(d-1)^2$ out of $d^2-1$ decay rates can be always negative.

The examples of ENM evolution given above all have two broad restrictions: (1) the concept of non-Markovianity here refers to CP-indivisibility, and not the stronger condition of P-indivisibility; (2) they consider only unital-- specifically Pauli-- channels. Here we wish to address the 2nd of these points. So far as known to us, there are only a few works addressing similar questions about non-unital channels. Here we undertake to do this, studying in detail the generalized amplitude damping (GAD) channels \cite{nielsen2002quantum,srikanth2008squeezed}. 

Before embarking on our results, we present a notion of non-unital non-Markovianity as below.
Any $d$-dimensional dynamical map may be given a matrix representation as $\mathcal{E}(t) \rightarrow F(t)$ as 
\begin{equation}
	F(t)=\left(\begin{array}{c|c}
		1 & {\bm 0}_{1\times(d^{2}-1)}\\
		\hline \bm{\tau} & M
	\end{array}\right),
	\label{eq:mapmatrix}
\end{equation}
where $\bm{\tau} \in \mathbb{R}^{d^2 -1}$ and $M$ is a $d^2-1 \times d^2 -1$ real matrix. Given the state $\rho=\frac{1}{d} (\mathbb{I} + \sum_{i=1}^{d^2-1} r_i G_i $), where $G_i$ are orthonormal basis given by the generalized Gell-Mann matrices with $G_0 = \frac{\mathbb{I}}{\sqrt{d}}$ and the Hermitian $G_i$ where $i \in \{1,.....d^2 -1\}$, then the map $F_{ij}(t) := {\rm Tr}(G_i \mathcal{E}[G_j])$, which transforms a Bloch vector $\bm{r}$ as
\begin{align}
	\bm{r} \rightarrow \bm{r}^\prime(t) = M(t)\bm{r}(0) + \bm{\tau}(t). \label{eq:blochmap}
\end{align}
If $\bm{r_1}$ and $\bm{r_2}$ are two Bloch vectors corresponding to the states $\rho_1$ and $\rho_2$, then from Eq.~(\ref{eq:mapmatrix}) and (\ref{eq:blochmap}), the transformation $\mathcal{E}(t)[\rho_1 - \rho_2]$ corresponds to $M(t)[\bm{r_1} - \bm{r_2}]$. Therefore, the trace distance fails to witness non-Markovianity originating solely from $\bm{\tau}$. However, the full non-Markovianity, in the sense of P-indivisibility is encoded in the map $F(t)$. Interestingly, the divisibility property of $\mathcal{E}(t)$ carried over as $F(t) = F(t,s)F(s)$ leads to a non-trivial relationship between unital and non-unital parts of the channel:
\begin{align}
	M(t) = M(t,s) M(s) \quad \text{and} \quad \bm{\tau}(t) = \bm{\tau}(t,s) + M(t,s) \bm{\tau}(s).
	\label{eq:nontrivial}  
\end{align}
Therefore, the full P-indivisibility is captured by Eq.~(\ref{eq:nontrivial}).
\begin{definition}
	When non-Markovianity originates solely from the non-unital part, i.e., from $\bm{\tau}(t)$, then we call it \textit{non-unital non-Markovianity} and such a channel will be called purely non-unital non-Markovian channel.\label{def:def1}
\end{definition}Given this notion, we ask: is eternal non-unital non-Markovianity, as per the Definition~(\ref{def:def1}), possible? We answer this in the negative for a class of non-unital channels. However, we will later show that for a class of phase-covariant channels and higher dimensional non-unital channels, eternal non-Markovianity is still possible, nonetheless such a non-Markovianity originates from the unital part of the channel. 

This is work is organized as follows. In Section \ref{sec:impossible} we introduce a $d$-level GAD channel for $d$-dimensional systems and show that it cannot be eternally non-Markovian. In Section \ref{sec:utmost}, we further provide rigorous proof for the case of qubit GAD channel and show that the utmost one can achieve in this direction is a \textit{quasi}-eternal non-Markovian GAD channel, of which we explicitly construct a qubit example. In Section \ref{sec:general}, we point out that this negative result does not extend to non-unital channels in general, by providing an example of a non-unital channel that is ENM. We present our conclusions in Section \ref{sec:conclude}.

\section{Generalized amplitude damping and the impossibility of eternal non-Markovianity \label{sec:impossible}} 
\subsection{Qudit generalized amplitude damping}
A well-known exemplar of a qubit non-unital channel is the GAD channel \cite{nielsen2002quantum,banerjee2008symmetry}.  In this work, generalizing the well-known $d$-level amplitude damping channel \cite{dutta2016entanglement,fonseca2019high}, we introduce a $d$-level GAD channel $\mathcal{E}_G$ for $d$-dimensional systems given by 
\begin{equation}
	\mathcal{E}_G = \sum_{l=0}^{d-1} p_l \mathcal{E}_l, 
	\label{eq:gad2}
\end{equation}
where $p_l$ is the probability with which the $l^{\rm th}$ damping channel $\mathcal{E}_l$ acts on a $d$-dimensional system, along with the constraint that $\sum_l p_l = 1$. Here $\mathcal{E}_l$ damps an arbitrary input state $\rho$ to the state $\ket{l}\!\bra{l}$ with probability $\lambda$, and is defined by the Kraus representation $\mathcal{E}_l[\rho]  = \sum_{j=0}^{d-1} E_{l,j} \rho E_{l,j}^\dagger$ 

with  the Kraus operators in terms of damping parameter $\lambda$ as  
\begin{align}
	E_{l,l} &= \ket{l}\!\bra{l} + \sqrt{1-\lambda}\sum_{\substack{k=0 \\ k\ne l}}^{d-1} \ket{k}\!\bra{k} ;\nonumber\\
	E_{l,j} &= \sqrt{\lambda}\ket{l}\!\bra{j},\quad \quad (l \ne j) . 
	\label{eq:gad1}
\end{align} 
A key point about the damping channels is that they have a unique fixed point $\mathfrak{F}$. For the amplitude damping channels $\mathcal{E}_l$ it is $ \ket{l}\!\bra{l}$, and for the higher-dimensional GAD channel the fixed point is $\mathfrak{F} \equiv \sum_{l=0}^{d-1} p_l \ket{l}\!\bra{l}$. 

We have the following result. 
\begin{Theorem}
	A non-unital GAD channel in any finite dimension $d$ that is eternally CP-indivisible in the non-unital part is impossible.
	\label{thm:cpindivD}
\end{Theorem}
\begin{proof}
	In the Markovian case, any initial state is driven monotonically towards the fixed point $\mathfrak{F}$, fulfilling Eq. (\ref{eq:CPindiv}) for all pair of states $\rho_1$ and $\rho_2$.
	If this CP-divisibility condition is violated at time $t=0$, then there is a pair of states such that 
	\begin{align}
		\frac{d}{dt}\Vert [\mathbb{I} \otimes \mathcal{E}(t)] (\rho_1-\rho_2) \Vert_{ t=0} > 0.
		\label{eq:CPindivreq}
	\end{align}
	With the idea of maximizing the l.h.s in Eq. (\ref{eq:CPindivreq}), we let $\rho_2 \equiv \mathbb{I} \otimes \mathfrak{F}$  and $\rho_1 \equiv \mathbb{I} \otimes \rho_S$ for an arbitrary system state $\rho_S$. 
	
	It can be shown that the action of the GAD channel Eq. (\ref{eq:gad2}) is described by
	\begin{equation}
		\mathcal{E}(\rho_S)[t] = \lambda(t) \mathfrak{F} + \{(1-\lambda(t))\mathfrak{F} + \sqrt{1-\lambda(t)}\mathfrak{G}\},
		\label{eq:Egad}
	\end{equation}
	where $\mathfrak{G}$ is a matrix consisting only of off-diagonal terms dependent on the initial state. In the Markovian case, $\lambda(t)$ increases monotonically from 0 to 1, and thus $\mathcal{E}(\rho_S)[t]$ evolves monotonically from the initial state to $\mathfrak{F}$. Note that this corresponds manifestly to non-unital behavior. In specific, $\mathcal{E}(\mathfrak{F}) = \mathfrak{F}$ at all times $t$, as required by the definition of a fixed point. We have from Eq. (\ref{eq:Egad}) that
	\begin{equation}
		\Vert\mathcal{E}(\rho_S)[t] - \mathfrak{F}\Vert = \sqrt{1-\lambda(t)} \Vert\mathfrak{G}\Vert.
		\label{eq:sqrt}
	\end{equation}
	Given that GAD constitutes a CP map, $\lambda(t) \ge 0$ such that $\lambda(0)=0$. Thus, we have that $\dot{\lambda}|_{t=0}\ge0$. Then it follows from Eq. (\ref{eq:sqrt}) that
	\begin{align}
		\frac{d}{dt}\Vert \mathcal{E}({\rho}_S) - \mathfrak{F} \Vert_{ t=0} \le 0,
		\label{eq:thm}
	\end{align}
	contrary to the implication of Eq. (\ref{eq:CPindivreq}).
\end{proof}
In the following, the above result for the general case of qudits will be elucidated in the case of qubits. 
\subsection{Case of qubits}
In the case of qubits, we can illustrate the above results for maps in terms of implications for the decoherence rates. 
\begin{align}
	E_{1}&=\sqrt{1-p}\left[\begin{array}{cc}
		1 & 0\\
		0 & \sqrt{1-\lambda}
	\end{array}\right] ; \; E_{2} =\sqrt{1-p} \left[\begin{array}{cc}
		0 & \sqrt{\lambda}\\
		0 & 0
	\end{array}\right]; \nonumber \\
	E_{3}&=\sqrt{p}\left[\begin{array}{cc}
		\sqrt{1-\lambda} & 0\\
		0 & 1
	\end{array}\right] ; \; E_{4} =\sqrt{p}\left[\begin{array}{cc}
		0 & 0\\
		\sqrt{\lambda} & 0
	\end{array}\right],
	\label{eq:kraus_operators}
\end{align}
where $\lambda, p \in [0,1]$.  

A key point is that the form Eq. (\ref{eq:kraus_operators}) constrains the Markovian behavior of the channel, even without knowing details of the time dependence of the parameters.

\begin{lemma}
	\label{lem:one}
	The damping parameter $\lambda(t)$ of the GAD channel is such that
	\begin{equation}
		\lim_{t \rightarrow 0+} \dot{\lambda}\ge0,
		\label{eq:llimit}
	\end{equation}
	meaning that for sufficiently small time $t$, the variable $\lambda(t)$ cannot decrease.
\end{lemma}
\begin{proof}
	The functional form of Eq. (\ref{eq:kraus_operators}) implies that 
	\begin{align}
		\lambda(t=0) = 0 \;\; \text{and} \;\; 
		\forall_{t\ge0} ~\lambda(t) \ge 0,
		\label{eq:limit}
	\end{align}
	where the first equation is a consequence of requiring that a channel must reduce to the identity map $I$ at $t=0$, and the second that of requiring that the map specified by Eq. (\ref{eq:kraus_operators}) must be complete. 
	Eq. (\ref{eq:limit})  implies Eq. (\ref{eq:llimit}).
\end{proof} 
\color{black}

It is well known that trace-distance can witness P-indivisibility. Our first result is to show that GAD cannot manifest eternal P-indivisibility. To this end, for the above map we have the affine representation  
\begin{align}
	\vec{r}^\prime (t) = M(t) \vec{r}(0) + \vec{\tau}(t)
	\label{eq:affine-rep}
\end{align}where $M$ is a $3 \times 3$ real matrix which modifies the Bloch vector $\vec{r} = (r_1,r_2,r_3)^T$, and $\vec{\tau} = (\tau_1,\tau_2,\tau_3)^T$ is the shift vector representing the displacement of Bloch vector towards the fixed point of the non-unital channel under consideration. Now, for a GAD channel we have
\begin{align}
	M &= \left(
	\begin{array}{ccc}
		\sqrt{1-\lambda} & 0 & 0 \\
		0 & \sqrt{1-\lambda} & 0 \\
		0 & 0 & 1-\lambda \\
	\end{array}
	\right); \quad
	\vec{\tau} =\left(
	\begin{array}{c}
		0 \\
		0 \\
		(1-2p)\lambda\\
	\end{array}
	\right).
	\label{eq:affine_GAD}
\end{align} 
The non-Markovianity can arise from the unital part (the matrix $M$) or the non-unital part (the shift $\vec{\tau}$).
For the most part in this paper, we set the notation $p(t)=p$ and $\lambda(t)=\lambda$ for simplicity.

The block diagonal representation of a dynamical map provides a straightforward way of connecting trace distance with the decay rates \cite{hall2014canonical}, hence we make use of this feature here.  Eq. (\ref{eq:affine-rep}) obeys the canonical master equation in the Block diagonal representation
\begin{align}
	\frac{d\vec{r}(t)}{dt} = D(t) \vec{r}(t) + \vec{\mu}(t)
\end{align}
where $D(t)$ and $\vec{\mu}(t)$ are called damping matrix and drift vector respectively. In relation with Eq. (\ref{eq:affine-rep}) 
we have $
\frac{dM(t)}{dt} = D(t) M(t) \quad \implies \quad M(t) = e^{\int_{t_0}^{t}D(s) ds}.
$
The channel $\mathcal{E}$ is CP-divisible when $\text{Tr}[D(t)] < 0$ and P-divisible when $h_{\rm max}(D(t)+D(t)^{\rm T})< 0$ \cite{hall2014canonical}, where $h_{\rm max}(X)$ is the largest eigenvalue of the matrix $X$.

\begin{Theorem}
	A qubit GAD channel that is eternally non-Markovian as witnessed by the trace-distance measure is impossible.
	\label{thm:pindiv}
\end{Theorem}
\begin{proof}
	For the qubit GAD channel Eq.~(\ref{eq:kraus_operators}), we have:
	\begin{equation}
		D(t) = \begin{pmatrix}
			\frac{-\dot{\lambda}}{2(1-\lambda)} & 0 & 0 \\
			0 & \frac{-\dot{\lambda}}{2(1-\lambda)} & 0 \\
			0 & 0 & \frac{-\dot{\lambda}}{1-\lambda} 
		\end{pmatrix}.
	\end{equation}
	At time $t \rightarrow 0$, we have $h_{\rm max}(D(t)+D(t)^{\rm T})= \frac{-\dot{\lambda}}{1-\lambda}$, which by virtue of Lemma \ref{lem:one} is negative.
\end{proof}

We also have that in the limit $t \rightarrow 0$, $\textrm{Tr}[D(t)] = \frac{-2\dot{\lambda}}{1-\lambda} < 0$ on account of Eq. (\ref{eq:llimit}). Therefore, the Bloch volume measure   \cite{lorenzo2013geometrical, hall2014canonical} also is unable to indicate an ENM property. 

One might suppose that even though the GAD channel cannot be eternally P-indivisible, it could still be eternally CP-indivisible. Our next result is that even this is ruled out, because the channel must start off as Markovian. 
\begin{Theorem}
	A qubit GAD channel that is eternally CP-indivisible is impossible.
	\label{thm:cpindiv}
\end{Theorem}
\begin{proof}
	For the qubit GAD channel, the canonical master equation is
	\begin{align}
		\dot{\rho}(t) &= \mathcal{L}[\rho(t)] \nonumber \\
		&= \gamma_1(t)[\sigma_{-} \rho(t) \sigma_{+} - \frac{1}{2} \{\sigma_{+}\sigma_{-},\rho(t)\}] \nonumber \\ 
		&+ \gamma_2(t)[\sigma_{+} \rho(t) \sigma_{-} - \frac{1}{2}\{\sigma_{-} \sigma_{+} ,\rho(t)\}]
		\label{eq:master}
	\end{align}
	where $\sigma_{+}=\ket{1}\bra{0}$ and $\sigma_{-}=\ket{0}\bra{1}$ are the raising and lowering operators. The canonical decay rates are given by:
	\begin{align}
		\gamma_1(t) &= \lambda  \dot{p}+\frac{\dot{\lambda} p}{1-\lambda},
		\nonumber \\
		\gamma_2(t) &= \lambda  \dot{q}+\frac{\dot{\lambda} q}{1-\lambda},
		\label{eq:ratesgeneric}
	\end{align}
	where $p+q=1$. We thus have $\gamma_1 + \gamma_2 = \frac{\dot{\lambda}}{1-\lambda}$ in as much as $\dot{p} = -\dot{q}$. Eqs. (\ref{eq:ratesgeneric}) entail that 
	\begin{align}
		\lim_{t \rightarrow 0+} \gamma_1(t) &= \dot{\lambda} p,
		\nonumber \\
		\lim_{t \rightarrow 0+} \gamma_2(t) &= \dot{\lambda} (1-p).
		\label{eq:ratesgeneric+}
	\end{align}
	It then follows by virtue of Eqs. (\ref{eq:llimit}) and (\ref{eq:ratesgeneric+}), and noting that $p$ is a probability, that $\lim_{t \rightarrow 0+} \gamma_j \ge0$ ($j=1,2$).  
\end{proof}
\bigskip

\section{Quasi-eternally non-Markovian GAD channel \label{sec:utmost}}

Given the impossibility of an ENM qubit GAD channel, the utmost that can be achieved in this direction is quasi-eternal non-Markovianity, whereby one of the rates becomes negative after a finite time $t_0$, and remains negative forever thereafter. The closest approximation to eternal non-Markovianity would then be by letting $t_0 \rightarrow 0$. An example of this sort is discussed below. To this end, we introduce an eternal ``unmixing'' effect, i.e., having $\dot{p}(t)<0$ for all times. In that spirit, in Eq. (\ref{eq:kraus_operators}), let
\begin{align}
	p(t) = \frac{e^{- m t}}{n} \;\; ; \;\; \lambda(t) = 1- e^{- \nu t} 
	\label{eq:forms}
\end{align}
where $m, n$ and $\nu$ are some real positive constants.
From Eq. (\ref{eq:ratesgeneric}), we find that  
\begin{subequations}
	\begin{align}
		\gamma_1(t) &=\frac{1}{n} e^{-m t} \left(m \left(e^{-\nu t}-1\right)+\nu \right) \label{eq:rate1} \\
		\gamma_2(t) &=\nu-\gamma_1(t) \label{eq:rate2}
	\end{align}
	\label{eq:rate12}
\end{subequations}
are the canonical decay rates, which are plotted in Figure (\ref{fig:rates}) for the values $m=3,n=2$ and $\nu=1$.   

Decoherence rate $\gamma_2$ remains positive throughout. The initial time when rate $\gamma_1$ turns negative is found using Eq. (\ref{eq:rate1}) to be 
\begin{equation}
	t^\ast = \frac{1}{\nu} \log\left(\frac{m}{ m - \nu }\right).
	\label{eq:tast}
\end{equation}
Note that since $m>\nu$, thus there is an initial interval given by $0 \le t < t^\ast$, during which $\gamma_1$ is necessarily positive. In Eq. (\ref{eq:rate1}), the term $(m(e^{-\nu t}-1)+\nu) \rightarrow \nu - m$, in the limit $t \rightarrow \infty$. Therefore the decay rate approaches $0^{-}$, provided $m >\nu $. This provides the condition for the quasi-eternal non-Markovianity of the channel. The symmetry between the two rates is a consequence of Eq. (\ref{eq:rate2}), and more generally, Eq. (\ref{eq:ratesgeneric}), and is reflected in Figure \ref{fig:rates}.

\begin{figure}[ht!]
	\centering
	\includegraphics[width=0.8\textwidth]{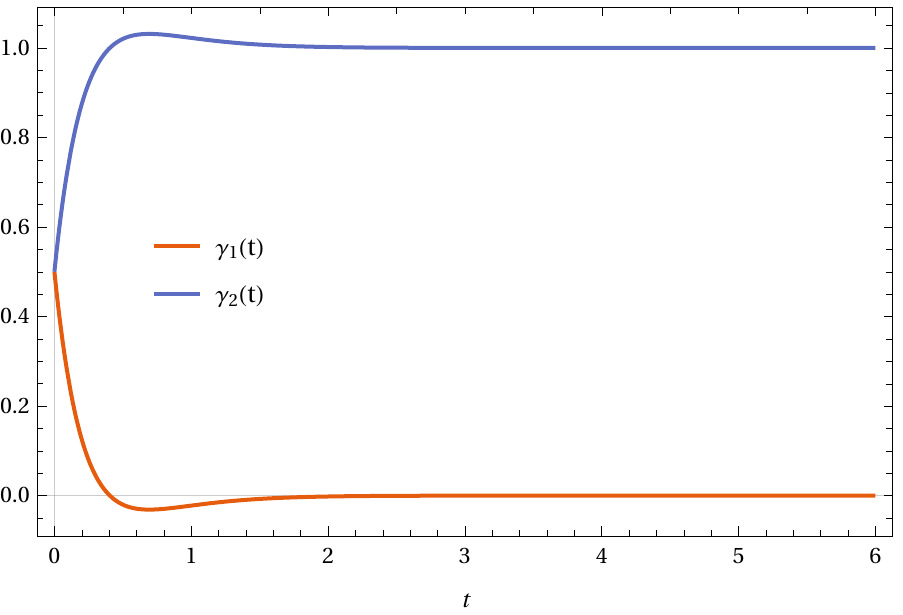}
	\caption{(Color online) The decay rates for the quasi-eternal non-Markovian GAD channel described by the Lindblad rates Eq. (\ref{eq:rate12}), with values $m:=3, n:= 2$ and $\nu:=1$. Whilst rate $\gamma_2$ is always positive, rate $\gamma_1$ is positive until time $t^\ast$, Eq. (\ref{eq:tast}), before becoming negative, and thereafter remains negative throughout thereafter. It is a consequence of Eq. (\ref{eq:tast}) that $t^\ast$ can be made arbitrarily close to 0, but not 0 itself, by virtue of Theorem 1. The symmetry between the two decay rates is a consequence of the fact, which follows from Eqs. (\ref{eq:rate1}) and \ref{eq:rate2}), that $\gamma_1(t)+\gamma_2(t) = \nu$.}
	\label{fig:rates}
\end{figure}
This channel by design is P-divisible in the sense that the intermediate map $F(t,s)$, according to Eq.~(\ref{eq:nontrivial}), is positive at all times and thus will be indicated to be non-Markovian neither by a distinguishability criterion such as trace distance \cite{liu2013nonunital}, nor by a criterion based on entropic quantities \cite{megier2021entropic}. On the other hand, the occurrence of negative $\gamma_1$ implies that it is CP-indivisible. The non-Markovianity of the channel may be quantified using a measure due to Hall-Cresser-Li-Anderson \cite{hall2014canonical}, which is given by $\mathcal{\xi}_{\rm \small HCLA} := - \int_{t : \gamma(t) < 0 }  \gamma(t) dt $.
In the context of this paper, $\gamma_1(t) < 0$ for all $t > t^\ast$, hence we find that
\begin{align}
	\mathcal{\xi}_{\rm \small HCLA} &:= - \int_{{t^\ast} }^{\infty} \gamma_1(t) dt  \nonumber \\
	&= \frac{\nu}{(\nu + m)n}\left(-\frac{m}{\nu - m}\right)^{-\frac{\nu + m}{\nu}}.
\end{align}
where $t^\ast$ is given in by Eq. (\ref{eq:tast}). The HCLA measure is plotted in Fig. (\ref{fig:hcla}).

\begin{figure}
	\centering
	\includegraphics[width=0.8\textwidth]{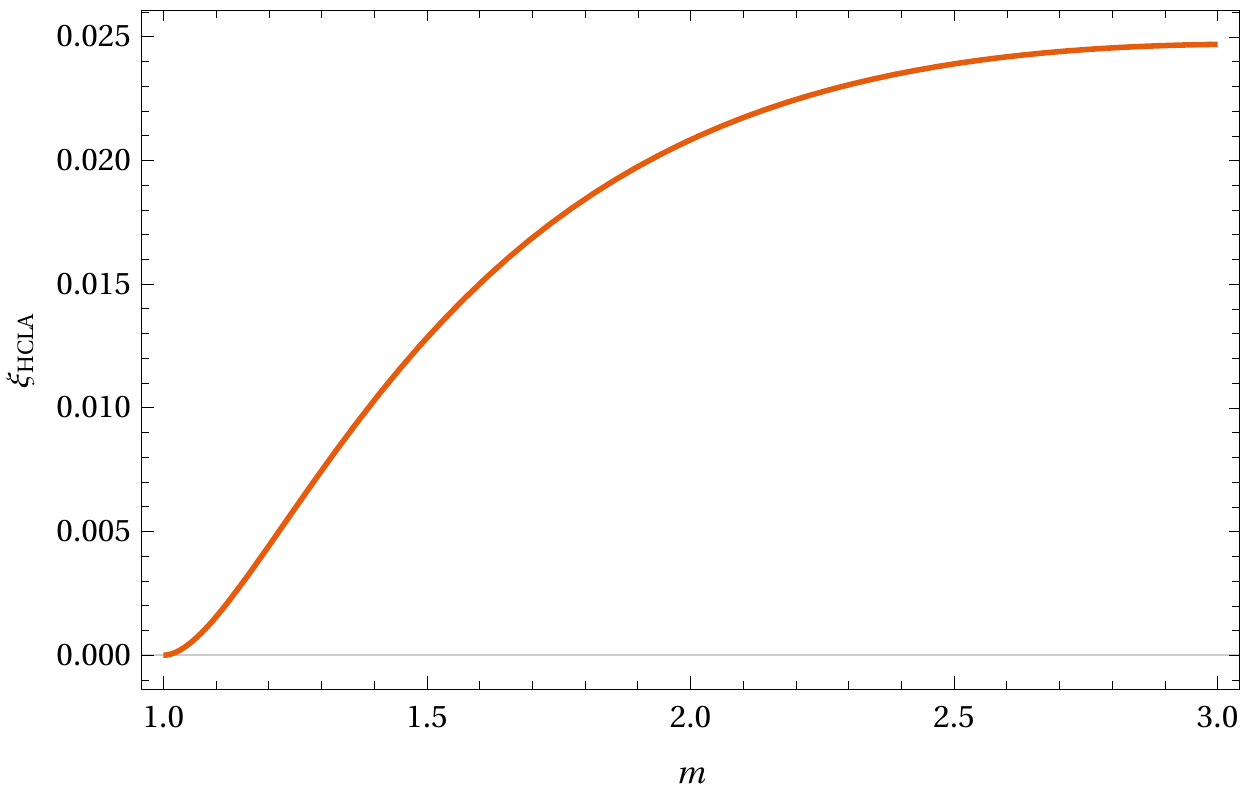}
	\caption{(Color online) The HCLA measure of QENMGAD channel (\ref{eq:kraus_operators}) against $m$.}
	\label{fig:hcla}
\end{figure}

\section{Non-unital channels and eternal non-Markovianity \label{sec:general}}
The impossibility of ENMity of GAD provokes the question of whether this is generic for all qubit non-unital channels. Yet, a simple example suffices to show that this is not the case. Consider the non-unital channel generated by linearly combining the generators the above Pauli ENM channel and an amplitude damping (AD) channel. An example would be:
\begin{align}
	\mathcal{L}(\rho) &= [\sigma_X\rho \sigma_X - \rho] + [\sigma_Y\rho \sigma_Y - \rho]  - \tanh(t)[\sigma_Z\rho \sigma_Z - \rho] \nonumber \\
	&+ \gamma(t)[\sigma_- \rho \sigma_+ - \frac{1}{2}\{\sigma_+ \sigma_-, \rho\}],
	\label{eq:non-unital}
\end{align}
where $\gamma(t)$ is arbitrary decoherence rate for an AD channel. The action of a unital channel leaves the identity operator unchanged, whence $\dot{\rho} = \mathcal{L}_{\rm unital}(\rho)=0$ if $\rho \propto \mathbb{I}$. The non-unitality of the channel described in Eq. (\ref{eq:non-unital}) follows,  noting that $L(\mathbb{I}) \ne 0$.
Employing the notation $\sigma_\mp \equiv \frac{1}{{2}}(\sigma_X \pm i \sigma_Y)$, the above equation can be rewritten as
\begin{align}
	\mathcal{L}(\rho) = \sum_{j,k} d_{j,k} (\sigma_j \rho \sigma_k - \frac{1}{2}\{\sigma_k\sigma_j,\rho\}),
\end{align}
where the decoherence matrix ${\bf d} \equiv \{d_{j,k}\}$ \cite{hall2014canonical} is given by
\begin{equation}
	{\bf d} = \begin{pmatrix}
		1+\frac{\gamma}{2} & -i\frac{\gamma}{4} & 0 \\
		i\frac{\gamma}{4} & 1+\frac{\gamma}{2} & 0 \\
		0 & 0 & -\tanh(t)
	\end{pmatrix}
\end{equation}
The eigenvalues of {\bf d}, which correspond to the canonical rates, are found to be $1+\gamma, 1,-\tanh(t)$. The last eigenvalue obviously entails ENM-ity. Note that this example can be readily extended by using the GAD channel instead of the AD channel. 

Ref. \cite{filippov2020phase} points out an example of non-unital ENM channel  belonging to the class of phase covariant qubit channels $\mathcal{E}_{\rm pc}$, which are characterized by the property that $e^{i\mathcal{E} \sigma_z}\mathcal{E}_{\rm pc}[\rho]e^{-i\mathcal{E}\sigma_z} = \mathcal{E}_{\rm pc}[e^{i\mathcal{E} \sigma_z}\rho e^{-i\mathcal{E}\sigma_z}]$.  The generator has the form 
$
\mathcal{L}_{\rm pc} \equiv \gamma_z (t)[\sigma_z\rho \sigma_z - \rho] 
+ \gamma(t)[\sigma_+ \rho \sigma_- - \frac{1}{2}\{\sigma_- \sigma_+, \rho\}]
+ \gamma(t)[\sigma_- \rho \sigma_+ - \frac{1}{2}\{\sigma_+ \sigma_-, \rho\}],
$
where it can be arranged so that $\gamma_z < 0$ for all $t > 0$ \cite{filippov2020phase}. Incidentally, this model was also considered in Ref. \cite{vacchini2011markovianity} in a different context.

Although this settles the question regarding the existence of  a purely  non-unital ENM channel, this also leads to the observation that in the above mentioned instances, the channel is not purely non-unital, in that the generator in either case can be expressed as a mixture that includes a unital component. Moreover, it is this unital component that introduces the ENM-ity.

It could be asked whether this no-go result is specific to two-level systems. It is straightforward to show that these results hold in general to arbitrary finite dimensions, by taking a trivial extension of the qubit dynamics, by embedding the qubit in two dimensions and padding the additional dimensions with an identity operation.
This point can be manifestly illustrated by considering another simple way to produce an example of non-unital ENM-ity, namely by combining the Pauli ENM channel with non-unital channels taking recourse to a higher dimensional space. Consider a ququart (4-level quantum system), wherein the  ``lower'' two-dimensional subspace, spanned by $\ket{0}$ and $\ket{1}$, is subject to the above Pauli ENM channel, whilst its ``higher'' two-dimensional subspace, spanned by $\ket{2}$ and $\ket{3}$, is subject to an AD or GAD channel. In other words, the two parts of the ququart are locally decohered, so that the sum of diagonals is preserved within each subspace and there are no cross-terms in the noise operator across the two subspaces (e.g., $\ket{1}\bra{2}$, etc.). This is mathematically captured by requiring that the generator for the ququart evolution is given by:
\begin{equation}
	\mathcal{L}_4 = \mathcal{L}_{01}(t)\oplus \mathcal{L}_{23}(t),
	\label{eq:L4}
\end{equation}
where the subscripts to the symbol $\mathcal{L}$ denote the subspaces acted upon. Evidently, the ququart as a whole is subject to non-unital noise, which is ENM by virtue of the evolution in first qubit subspace. The resulting map is given is given by
\begin{align}
	e^{\int_0^t \mathcal{L}_4(s) ds} [\rho] &= \begin{pmatrix}
		\mathcal{E}_{01} & \textbf{0} \\\textbf{0}  & \mathcal{E}_{23} \nonumber \\
	\end{pmatrix}[\rho]
	\nonumber \\
	&= \begin{pmatrix}
		\mathcal{E}_{01} & \textbf{0} \\\textbf{0}  & \mathbb{I}_2\nonumber \\
	\end{pmatrix} \cdot \begin{pmatrix}
		\mathbb{I}_2 & \textbf{0} \\\textbf{0}  & \mathcal{E}_{23} \nonumber \\
	\end{pmatrix}[\rho]
	\nonumber \\
	& = \sum_{j = 1}^3 \sum_{l = 1,2} M_j N_l \rho N_l^\dagger M_j ^\dagger
\end{align}
where \textbf{0} represents a $2 \times 2$ null matrix. Here the Kraus operators $M_j$ for the Pauli channel are obtained by padding an identity operator: $M_j = (E_j \oplus \mathbb{I}_2)$, where $E_j$ are given by Eq. (\ref{eq:ENM}),  In particular,  $M_1 = \sqrt{2k} \mathbb{I}_4$, and
$$
M_2 = \begin{pmatrix}
	0 & \sqrt{k} & 0 & 0 \\
	\sqrt{k}& 0 & 0 & 0 \\
	0 & 0 & 1 & 0 \\
	0 & 0 & 0 & 1
\end{pmatrix};
~~
M_3 = \begin{pmatrix}
	0 & -i\sqrt{k} & 0 & 0 \\
	i\sqrt{k} & 0 & 0 & 0 \\
	0 & 0 & 1 & 0 \\
	0 & 0 & 0 & 1
\end{pmatrix};
$$
The Kraus operators for the AD channel similarly are given by: $N_k \propto (\mathbb{I}_2 \oplus K_k)$, where $K_k$ are the qubit Kraus operators given in Eq. (\ref{eq:kraus_operators}), with $p{:=}1$. 
$$
N_1 = \begin{pmatrix}
	1 & 0 & 0 & 0 \\
	0 & 1 & 0 & 0 \\
	0 & 0 & 1 & 0 \\
	0 & 0 & 0 & \sqrt{1-\lambda}
\end{pmatrix};
~~
N_2 = \begin{pmatrix}
	1 & 0 & 0 & 0 \\
	0 & 1 & 0 & 0 \\
	0 & 0 & 0 & \sqrt{\lambda} \\
	0 & 0 & 0 & 0
\end{pmatrix}
$$
the ququart Kraus operators for the ENM and AD channels, respectively, where $k$ is given by Eq. (\ref{eq:k}). 
These diverse examples examined in the Section suggest that a purely non-unital channel (i.e., one lacking
any unital component) is impossible. Although we have only studied the qubit case, we conjecture that it is true for non-unital channels in arbitrary finite dimensions.

\section{Conclusions \label{sec:conclude}}
The ENM Pauli channel \cite{hall2014canonical} presented the rather counter-intuitive idea of a quantum noisy channel that is \textit{recoherent} for all $t>0$. Here we address the question of the possibility of a non-unital channel that is ENM, and show that it is impossible for the well-known class of GAD channels of qubits. 

We point out that this impossibility of eternal non-Markovianity does not hold for arbitrary non-unital channels. Specifically, as evident from examples discussed above, one can construct an ENM non-unital channel by a linear combination of the generators of an ENM unital channel and an arbitrary non-unital one. However, we conjecture that for a non-unital channel with no such unital component (in the above sense) to be ENM, is impossible.

A consequence of our above negative result is that the best that a qubit GAD channel can approach the ENM property is to be non-Markovian for all times above a non-zero threshold time $t^{\ast}$, i.e., strictly $t^{\ast} > 0$. By way of illustration, we provided an example of a ``quasi-ENM''  GAD channel, which is non-Markovian for all time $t \ge t^\ast > 0$. 

\section*{Acknowledgments}
SU was supported by IIT Madras via the Institute Postdoctoral Fellowship. R.S.   acknowledges the support of Department of Science and Technology (DST), India, Grant No.: MTR/2019/001516. RS and SB acknowledge support from Interdisciplinary Cyber Physical Systems (ICPS) programme of the Department of Science and Technology (DST), India, Grant No.: DST/ICPS/QuST/Theme-1/2019/6 and DST/ICPS/QuST/Theme-1/2019/13. SB also acknowledges support from the Interdisciplinary Program (IDRP) on Quantum Information and Computation (QIC) at IIT Jodhpur.

\bibliographystyle{unsrt}
\bibliography{eternal}
\end{document}